\documentclass[letterpaper,10pt,conference]{ieeeconf}

\IEEEoverridecommandlockouts
\usepackage{amsmath,amssymb,amsfonts,bm,cite}

\usepackage{amsthm}

\usepackage{url}

\usepackage{algorithm}
\usepackage{algpseudocode}
\usepackage{booktabs}
\usepackage{graphicx}
\usepackage{xcolor}
\usepackage{hyperref}
\usepackage{tikz}

\theoremstyle{plain}

\newtheorem{theorem}{Theorem}
\newtheorem{lemma}{Lemma}

\newtheorem{assumption}{Assumption}

\renewenvironment{proof}[1][Proof]
{\par\noindent\textit{#1.}\ }
{\hfill$\square$\par}

\newcommand{\R}{\mathbb{R}}
\newcommand{\calN}{\mathcal{N}}
\newcommand{\calE}{\mathcal{E}}
\newcommand{\calD}{\mathcal{D}}

\newcommand{\Id}{\bm I}
\newcommand{\one}{\bm 1}
\newcommand{\p}{\bm p}
\newcommand{\q}{\bm q}
\newcommand{\vvec}{\bm v}
\newcommand{\vref}{\bm v^{\rm ref}}
\newcommand{\K}{\bm K}
\newcommand{\kvec}{\bm k}
\newcommand{\Rmat}{\bm R}
\newcommand{\X}{\bm X}

\newcommand{\A}{\bm A}
\newcommand{\E}{\bm E}
\newcommand{\e}{\bm e}
\newcommand{\D}{\bm D}
\newcommand{\Chat}{\hat{\bm C}}

\newcommand{\Diag}{\operatorname{Diag}}

\newcommand{\proj}{\Pi}

\newcommand*\circled[1]{%
\tikz[baseline=(char.base)]{%
\node[shape=circle,draw,inner sep=0.05pt] (char) {#1};}}

\title{Decentralized Gain Learning for Voltage Control \\ via Local Trajectory Convolution}

\author{Yiwei Dong, Wenqi Cui 
\thanks{Yiwei Dong and Wenqi Cui are with the Department of Electrical and Computer Engineering, New York University, NY 11201, USA. (e-mail: yiweidong@nyu.edu; wenqicui@nyu.edu)}%
\thanks{This work is supported in part by  the Henry Luce Foundation.}}

\begin{document}
\maketitle
\thispagestyle{empty}
\pagestyle{empty}
\begin{abstract}
Fast and spatially heterogeneous fluctuations from distributed energy resources call for voltage regulation that is both fast and responsive to changing operating conditions.  Local linear Volt/VAR control offers the advantage of promptly regulating voltages without real-time communication, but its performance depends critically on the choice of control gains. Optimizing these gains for a system-level objective, however, typically requires an accurate network model or centralized communication. In this work, we uncover a structural property that enables such optimization using only local measurements. Specifically, by exploiting the self-adjoint structure of the closed-loop voltage dynamics, we show that each component of the system-level objective gradient can be exactly recovered from a self-convolution of the corresponding bus's voltage deviation trajectory. Building on this result, we develop a simple decentralized gain optimization method in which each bus updates its control gain using only its local data trajectory. Simulation results demonstrate that the proposed approach effectively updates the Volt/VAR control gains in response to changing operating conditions while requiring only local information.


\end{abstract}

\section{Introduction}
\label{sec:introduction}
Voltage regulation in distribution systems is becoming increasingly challenging as operating conditions vary more rapidly and heterogeneously. High penetrations of distributed energy resources can cause local overvoltage~\cite{7244261}, while emerging large loads, such as AI data centers, can introduce rapid power fluctuations~\cite{chen2026electricitydemandgridimpacts}.
These changes require voltage control actions faster than those provided by conventional tap changing transformers and switched capacitors~\cite{lundberg2022local}. At the same time, inverter-based resources (e.g., battery energy storage, solar photovoltaics) provide fast reactive power actuation, making them well suited for distribution system voltage regulation.

A central question is how to coordinate a large number of inverter-based resources to maintain satisfactory voltage profiles across the distribution network. Centralized approaches typically formulate voltage regulation as an optimal power flow problem~\cite{colot2024optimal,dong2026datadrivensuccessivelinearizationoptimal} and achieve strong control performance under a broad range of operating conditions. However, their implementation generally relies on wide-area communication and (or) accurate knowledge of the network model, both of which are difficult to obtain in many distribution systems. Distributed methods reduce the burden of centralized computation by coordinating bus-level reactive power decisions through local computation and communication among neighboring buses~\cite{6930779,8779692,8667359}. Nevertheless, even this level of real-time neighbor-to-neighbor communication can be difficult to support reliably in practical distribution networks.

To enable voltage regulation without relying on communication infrastructure, local Volt/VAR control has been widely studied~\cite{6666945,6760555,zhu2015fast,9091863,lundberg2022local}. A representative approach is droop control, in which each inverter adjusts its reactive power output linearly in response to locally measured voltage deviations. Voltage stability can be ensured by appropriately restricting the controller gains~\cite{6760555,zhu2015fast,9091863}. However, the resulting control performance, measured in terms of cumulative voltage deviations and control effort, depends critically on how these gains are selected. To address this issue, linear-quadratic-regulator-based formulations have been developed to optimize the gains with respect to quadratic performance objectives~\cite{lou2018optimal,jing2025linear}. More recently, data-driven and reinforcement-learning-based approaches have been proposed to reduce dependence on an accurate network model~\cite{10158325,10144580,10336939}. Nevertheless, despite the purely local parameterization of the controllers, the associated optimization or learning procedures typically still require communication or system-wide measurements. Consequently, controller gains are often tuned offline and held fixed during online operation. Such fixed gains may lead to degraded performance when operating conditions vary substantially because of time-varying distributed energy resources and loads~\cite{9963666}.

Recently, decentralized learning has been explored as a means of enabling model-free controller optimization without relying on wide-area communication~\cite{pmlr-v120-furieri20a,pmlr-v120-li20c,10383735, cui2022decentralized, pmlr-v211-jiang23a, jin2024approximate}. Most existing results have been developed in the context of linear–quadratic regulator problems~\cite{pmlr-v120-furieri20a,pmlr-v120-li20c, 10383735}, where distributed learning algorithms update local controller parameters through information exchange among neighboring agents. Achieving fully decentralized learning is substantially more challenging, because the gradient of the global performance objective generally depends on system-wide dynamics and trajectories~\cite{cui2022decentralized}. Some studies have identified specific structural conditions under which exact decentralized learning can be performed without communication~\cite{pmlr-v211-jiang23a, jin2024approximate}. However, these conditions do not hold for the voltage control problem. Fundamentally, optimizing linear Volt/VAR control gains remains a nonconvex problem, and evaluating the corresponding gradients requires either global model information or communication across the network.

This paper focuses on decentralized gain update for linear Volt/VAR control in power distribution systems.
Surprisingly, we found that the structure of the linearized voltage model enables the gradient of a system-level objective with respect to the Volt/VAR control gains to be \textit{exactly} decomposed using only combinations of local trajectories. Our key observation is that, under the linearized distribution flow (LinDistFlow) model~\cite{zhu2015fast}, the transition matrix of the closed-loop voltage deviation dynamics is
self-adjoint in the metric induced by the inverse reactive power-to-voltage sensitivity matrix. Leveraging this structure, for a system-level voltage regulation objective, each bus’s gradient
component can be expressed as a self-convolution
of its own voltage deviation trajectory. 
This result leads to a simple yet effective decentralized gain update scheme, in which each bus updates its local control gain through online gradient descent steps computed solely from locally measured data trajectories.
The resulting scheme is model-free, communication-free, and continuously adaptive to time-varying operating conditions, while provably stabilizing the system under bounded disturbances.








\textbf{Notation.}
Throughout this manuscript, vectors are denoted by lower-case bold symbols, matrices by upper-case bold symbols, and scalars by unbolded symbols. A symbol with subscript $i$ denotes the $i$-th component of the corresponding vector, while $t$ denotes the time step.


\section{Network Model and Controller Design}
\label{sec:model}

Voltage control in a distribution system involves coordinating reactive power injections across the network to maintain bus voltage magnitudes close to a reference value while avoiding excessive control effort. The rationale of Volt/VAR control lies in the physical relationship between nodal power injections and voltages. In distribution systems, this relationship is commonly approximated by the LinDistFlow model. This section begins
with the introduction of LinDistFlow, and then uses it to establish the basis of our controller design.


\subsection{Model}
\label{subsec:lindistflow_model}

Consider a radial distribution system as a connected graph $(\{0\}\cup\calN,\calE)$, with $\calN:=\{1,\ldots,N\}$ denoting the nonslack buses and $\calE$ denoting lines. Bus $0$ is the slack bus whose voltage is regulated by the point of connection with the transmission system.
Let
$\vvec_t\in\R^N$ collect the voltage magnitudes at the nonslack buses at time
$t$, and let $\p_t,\q_t\in\R^N$ denote the corresponding active and reactive
power injections. Under the LinDistFlow
model~\cite{zhu2015fast}, it satisfies
\begin{equation}
\label{eq:lindistflow_affine}
    \vvec_t
    =
    v_0\one
    +
    \Rmat\p_t
    +
    \X\q_t,
\end{equation}
where $v_0$ is the slack bus voltage magnitude, and
$\Rmat,\X\in\R^{N\times N}$ are the voltage sensitivity matrices associated
with active and reactive power injections, respectively.

One of the key structural properties leveraged in this paper is that the sensitivity matrices $\Rmat$ and $\X$ are positive definite. Specifically,  the sensitivity matrices reflect both the electrical parameters and the
topology of the feeder, and can be written as
\begin{equation}
\label{eq:RX_graph_form}
    \Rmat
    =
    \Chat^{-\top}\D_r\Chat^{-1},
    \quad
    \X
    =
    \Chat^{-\top}\D_x\Chat^{-1},
\end{equation}
where $\Chat\in\R^{N\times N}$ is the incidence matrix on nonslack buses. Matrices  $\D_r
    :=
    \Diag(r_\ell),
    \D_x
    :=
    \Diag(x_\ell)$ collect the resistance $r_\ell$ and reactance $x_\ell$ of each line $\ell\in\calE$,
respectively. 
Since $r_\ell>0\text{ and } x_\ell>0$ for all $\ell\in\calE$, $\D_r\succ0\text{ and }\D_x\succ0$. In addition, $\Chat$ is nonsingular for a radial connected network. It follows from
\eqref{eq:RX_graph_form} that $\X
    =
    \Chat^{-\top}\D_x\Chat^{-1}
    =
    \X^\top
    \succ0$, and the same applies to $\Rmat$.


    
Equation~\eqref{eq:lindistflow_affine} expresses how active and reactive
power injections determine the voltage profile. Correspondingly, one of the most common approaches to regulate voltage is through adjusting the reactive power $\bm{q}_t$. 
Given a desired
voltage profile $\vref\in\R^N$, the voltage control problem is therefore to
adjust $\q_t$ so that $\vvec_t$ remains close to $\vref$, while limiting the required reactive power effort.

\subsection{Local Volt/VAR Feedback Controller}
\label{subsec:closed_loop_dynamics}

Directly optimizing the reactive power from
\eqref{eq:lindistflow_affine} would generally require
network information and feeder wide measurements. To avoid relying on such
real-time information exchange and network model, a widely adopted control law is linear incremental control~\cite{6666945,zhu2015fast}, which iteratively adjusts local reactive power injections through local voltage deviation as 
\begin{equation}
\label{eq:incremental_controller}
    \q_{t+1}
    =
    \q_t-\K\bm{\tilde{v}}_t,
    \quad
    \K
    :=
    \Diag(k_1,\ldots,k_N),
\end{equation}
where $\bm{\tilde{v}}_t := \vvec_t-\vref$ is the voltage deviation from its reference value.
The diagonal structure in the control gain $\K$ makes the controller locally implementable, namely, $q_{t+1,i}    =
    q_{t,i}    -    k_i\tilde{v}_{t,i}$ at bus $i$.
Once the gains are specified, the controller can be executed independently at every bus without communication.


\subsection{System Level Gain Optimization}
\label{subsec:gain_objective}
The feedback gains directly shape the closed-loop voltage dynamics and
therefore need to be selected carefully. Existing local Volt/VAR controllers
typically use fixed parameters selected offline according to model-based stability conditions~\cite{6760555,zhu2015fast,9091863}. However, overly conservative gains may lead to insufficient voltage regulation, whereas overly aggressive gains can induce excessively large control inputs. We aim to optimize the control gains to reduce voltage deviations while avoiding excessive control effort. 
Specifically, let $\kvec
    :=
    (k_1,\ldots,k_N)^\top$ collect the diagonal entries of $\K$. For a finite horizon $T$, we optimize the local gains according to the following objective
\begin{equation}
\label{eq:regularized_objective_main}
J(\kvec)
:=
\sum_{t=1}^{T}
\bm{\tilde{v}}_t^\top\X^{-1}\bm{\tilde{v}}_t
+
\kvec^\top\bm \Lambda\kvec,
\end{equation}
where $\bm \Lambda:=\frac12\operatorname{Diag}(\lambda_1,\ldots,\lambda_N)$ is the regularization matrix with $\lambda_i>0$ being the regularization coefficients that penalize large feedback gains for bus $i$. The first term in equation~\eqref{eq:regularized_objective_main} measures the network-wide bus
voltage deviations accumulated over the horizon and also approximates the line loss~\cite{8779692,9091863}.  Moreover, we define the admissible gain set $\calD
    :=
    \prod_{i=1}^{N}
    [k_i^{\min},k_i^{\max}]$, where the bounds reflect inverter reactive power capability limits and are selected to maintain a sufficient stability margin.  The gain optimization problem is then
\begin{equation}
\label{eq:general_gain_objective}
\begin{aligned}
    \min_{\kvec\in\calD}\quad
    & J(\kvec) \\
    \mathrm{s.t.}\quad
    & \eqref{eq:lindistflow_affine} \text{ and } \eqref{eq:incremental_controller} \text{ for } t=1,\cdots, T.
\end{aligned}
\end{equation}
The optimization~\eqref{eq:general_gain_objective} depends on
the complete network model by construction. Therefore, methods
that optimize local controller parameters generally rely on
centralized computation or system-level information~\cite{lou2018optimal,
jing2025linear,9992647,10158325,10144580}. Such requirements are difficult to support in most distribution systems.

The central result of
this paper is that, despite this networkwide coupling,~\eqref{eq:general_gain_objective}  can be solved through gradient descent where each component of the
objective gradient can be recovered exactly from the voltage trajectory and the current gain
measured at the corresponding bus. Next, Section~\ref{sec:controller} presents the main theorem and the decentralized gain update algorithm. Section~\ref{sec:gradient} proves
the theorem and establishes that the algorithm coincides with projected gradient descent for the system-level objective.

\section{Decentralized Gain Update}
\label{sec:controller}



\subsection{Decentralized Gradient Decomposition}

To design a gain update rule that can optimize the system-level objective, we analyze the closed-loop dynamics of network voltages under the local Volt/VAR control~\eqref{eq:incremental_controller}. Consider a short time interval during which the active-power injection remains constant at $\p_0$. Combining
\eqref{eq:lindistflow_affine} and
\eqref{eq:incremental_controller} yields
\begin{align*}
    \bm{\tilde{v}}_{t+1}
    &=
    v_0\one
    +
    \Rmat\p_0
    +
    \X\q_{t+1}
    -
    \vref
    \\
    &=
    v_0\one
    +
    \Rmat\p_0
    +
    \X\q_t
    -
    \vref
    -
    \X\K\bm{\tilde{v}}_t
    \\
    &=
    (\Id-\X\K)\bm{\tilde{v}}_t.
\end{align*}
Define the transition matrix $\A(\K):=\Id-\X\K$. Then the closed-loop voltage deviation dynamics is
\begin{equation}
\label{eq:trajectory}
    \bm{\tilde{v}}_t=\A(\K)^t\bm{\tilde{v}}_0
    \quad
    t=0,1,2,\ldots .
\end{equation}
The vector $\bm{\tilde{v}}_0$ in \eqref{eq:trajectory} depends on the initial operating point $(\p_0,\q_0)$, whereas $\A(\K)$
governs how this deviation evolves under the current feedback gain. Under this closed-loop dynamics, problem~\eqref{eq:general_gain_objective} can be written as
\begin{equation}
\label{eq:general_gain_objective2}
    \min_{\kvec\in\calD}
     J(\kvec)= \sum_{t=1}^{T}
    \left(
        \A(\K)^t\bm{\tilde{v}}_0
    \right)^\top
    \X^{-1}
    \left(
        \A(\K)^t\bm{\tilde{v}}_0
    \right)
    +
    \kvec^\top\bm \Lambda\kvec.
\end{equation}
If gradient information is available, $\bm{k}$ can be optimized through gradient descent. However, 
computing even
a single gradient component $\partial J(\bm k)/\partial k_i$ would generally
require access to the network sensitivity matrix $\X$ and feeder wide voltage information.
Interestingly, we find that the structural properties of $\A(\K)$ allow the gradient $\partial J(\bm k)/\partial k_i$ at each bus $i$ to be expressed exactly in terms of the convolution of its local voltage deviation trajectory together with its control gain $k_i$. This result enables each bus to compute its component of the global gradient using only local information. We first state the resulting local gradient expression and corresponding algorithm, while the underlying structural property and proof are deferred to Section~\ref{sec:gradient}.

The following theorem shows how the local gradient can be exactly represented using the convolution of locally observed trajectories.


\begin{theorem}
\label{thm:exact_local_gradient} Let $\{\tilde{\bm{v}}_t\}$ be the
closed-loop trajectory generated by the gains $\kvec$
from an arbitrary initial voltage deviation $\tilde{\bm{v}}_0$. Then, for every $i\in\calN$,
the gradient of $J(\kvec)$ in
\eqref{eq:regularized_objective_main} with respect to $k_i$ is
\begin{equation}
\label{eq:local_regularized_gradient_main}
    \frac{\partial J(\kvec)}{\partial k_i}
    =
    -2
    \sum_{t=1}^{T}
    \sum_{s=0}^{t-1}
    \tilde{v}_{s,i}\tilde{v}_{2t-1-s,i}
    +
    {\lambda_i} k_i.
\end{equation}
\end{theorem}

Theorem~\ref{thm:exact_local_gradient} reveals that a $2T$-sample local trajectory exactly represents the gradient component of a $T$-step system-level objective. The right hand side of
\eqref{eq:local_regularized_gradient_main} depends only on the voltage samples
collected at bus $i$ and its own gain $k_i$. Thus, a gradient component
of the system-level objective can be evaluated without direct access to the
network sensitivity matrix or measurements from other buses. 

\subsection{Algorithm}
This local
gradient representation in \eqref{eq:local_regularized_gradient_main} leads directly to a simple decentralized gain update
procedure. Let superscript $m=0,1,2,\ldots$ index the gain update rounds. During round $m$, each bus holds its current gain $k_i^{(m)}$ fixed while
collecting its local voltage trajectory. At the end of the round, it computes
an update direction and updates its gain. 

\begin{algorithm}[t]
\caption{Decentralized Gain Update at Bus $i$}
\label{alg:local_gain_update}
\begin{algorithmic}[1]
\Require Horizon $T$, step size $ \alpha_i>0$, regularization parameter
${\lambda_i}>0$, gain interval $[k_i^{\min},k_i^{\max}]$, and initial gain
$k_i^{(0)}$
\For{$m=0,1,2,\ldots$}
    \State Measure and store the initial  voltage deviation
    $\tilde{v}_{0,i}^{(m)}$
    \For{$t=0,\ldots,2T-2$}
        \State Apply the incremental Volt/VAR feedback law
        \begin{equation*}
        \label{eq:controller_round_local}
            q_{t+1,i}^{(m)}
            =
            q_{t,i}^{(m)}
            -
            k_i^{(m)}\tilde{v}_{t,i}^{(m)}
        \end{equation*}
        \State Measure and store $\tilde{v}_{t+1,i}^{(m)}$
    \EndFor
    \State Compute the local update direction
    \begin{equation}
    \label{eq:local_update_signal}
        d_i^{(m)}
        :=
        -2
        \sum_{t=1}^{T}
        \sum_{s=0}^{t-1}
        \tilde{v}_{s,i}^{(m)}
        \tilde{v}_{2t-1-s,i}^{(m)}
        +
        {\lambda_i} k_i^{(m)}
    \end{equation}
    \State Update the local gain
    \begin{equation}
    \label{eq:projected_gain_update}
        k_i^{(m+1)}
        =
        \proj_{[k_i^{\min},k_i^{\max}]}
        \left(
            k_i^{(m)}
            -
            \alpha_i d_i^{(m)}
        \right)
    \end{equation}
\EndFor
\end{algorithmic}
\end{algorithm}

The
complete procedure is summarized in
Algorithm~\ref{alg:local_gain_update}, which is executed independently at
each bus. The reactive power update follows the local incremental control introduced in~\eqref{eq:incremental_controller}. 
The update direction in
\eqref{eq:local_update_signal} is computed solely from voltage
measurements collected at the same bus and its current gain. The projection
in \eqref{eq:projected_gain_update} guarantees that the updated gain remains
within its prescribed interval. Thus, every step of the algorithm can be implemented using only
locally available information. It requires neither prior knowledge of the network
topology, line parameters, voltage sensitivity matrices nor communication
among buses. 

\section{Optimization Interpretation}
\label{sec:gradient}

This section shows the proof of Theorem~\ref{thm:exact_local_gradient} and the convergence of the decentralized gradient update  in Algorithm~\ref{alg:local_gain_update}. 

\subsection{Decentralized Gradient Derivation}

To reveal how the local gradient $\partial J(\bm k)/\partial k_i$ can be exactly represented using local trajectories, we start by deriving an explicit representation of the gain gradient using the chain rule applied to
\eqref{eq:regularized_objective_main}. For each local gain $k_i$, 
\begin{equation}
    \frac{\partial J(\kvec)}{\partial k_i}
    =
    2
    \sum_{t=1}^{T}
    \bm{\tilde{v}}_t^\top
    \X^{-1}
    \frac{\partial \bm{\tilde{v}}_t}{\partial k_i}
    +
    {\lambda_i} k_i.
    \label{eq:gradJ}
\end{equation}

    Thus, it remains to characterize how the voltage trajectory $\bm{\tilde{v}}_t$
depends on $k_i$.
 Let $\e_i$ denote the $i$-th standard basis vector and define $\E_i:=\e_i\e_i^\top$. Since $k_i$ appears only in the $i$-th diagonal entry of $\K$,
\begin{equation*}
    \frac{\partial\K}{\partial k_i}
    =
    \E_i,
    \quad
    \frac{\partial\A(\K)}{\partial k_i}
    =
    -\X\E_i.
\end{equation*}

This relation gives us a compact representation of $\partial \bm{\tilde{v}}_t/{\partial k_i}$ by differentiating $\bm{\tilde{v}}_t=\A(\K)^t\bm{\tilde{v}}_0$:
\begin{equation}
\label{eq:de_dki}
\begin{aligned}
    \frac{\partial\bm{\tilde{v}}_t}{\partial k_i}
    &=
    \frac{\partial\A(\K)^t}{\partial k_i}\bm{\tilde{v}}_0\\
   & =
    \sum_{s=0}^{t-1}
    \A(\K)^{t-1-s}
    \frac{\partial\A(\K)}{\partial k_i}
    \A(\K)^s\bm{\tilde{v}}_0
    \\
    &=
    -
    \sum_{s=0}^{t-1}
    \A(\K)^{t-1-s}
    \X\E_i\bm{\tilde{v}}_s,
    \quad
    t\geq1.
\end{aligned}
\end{equation}
The second equality follows from the product rule for matrix powers,
and the last equality uses
$\A(\K)^s\bm{\tilde{v}}_0=\bm{\tilde{v}}_s$. 
Substituting~\eqref{eq:de_dki} into~\eqref{eq:gradJ} yields an explicit expression for each gradient component. However, this expression still requires the system information encoded in $\bm{X}$  and voltage measurements $\bm{\tilde{v}}$ across the network.
The following lemma establishes a property that  further decouples the voltage-related gradient term.

\begin{lemma}
\label{lem:self_adjoint_power}
 The
transition matrix satisfies $\A(\K)^\top\X^{-1}
    =
    \X^{-1}\A(\K)$. Moreover, for every nonnegative integer $r$,
\begin{equation}
\label{eq:self_adjoint_power}
    (\A(\K)^r)^\top\X^{-1}
    =
    \X^{-1}\A(\K)^r,
    \quad
    r\in\mathbb{Z}_{\geq0}.
\end{equation}
\end{lemma}
\begin{proof} We first show that the transition matrix is self-adjoint
under the inner product induced by $\X^{-1}$. Since $\X=\X^\top$, $\K=\K^\top$, and $\A(\K)=\Id-\X\K$, we have
\begin{equation}
\label{eq:self_adjoint_A}
\begin{aligned}
    \A(\K)^\top\X^{-1}
    &=
    (\Id-\K\X)\X^{-1}
    =
    \X^{-1}-\K
    \\
    &=
    \X^{-1}(\Id-\X\K)
    =
    \X^{-1}\A(\K).
\end{aligned}
\end{equation}

This self-adjoint property extends to every nonnegative integer
power of $\A(\K)$, and we prove this by induction. For $r=0$, $(\A(\K)^0)^\top\X^{-1}
    =
    \X^{-1}
    =
    \X^{-1}\A(\K)^0$. Suppose that $(\A(\K)^r)^\top\X^{-1}
    =
    \X^{-1}\A(\K)^r$ holds for some $r\in\mathbb{Z}_{\geq0}$. Then $(\A(\K)^{r+1})^\top\X^{-1}
    \stackrel{\circled{1}}{=}(\A(\K)^r\A(\K))^\top\X^{-1} \stackrel{\circled{2}}{=} \A(\K)^\top(\A(\K)^r)^\top\X^{-1} \stackrel{\circled{3}}{=}
    \A(\K)^\top\X^{-1}\A(\K)^r
    \stackrel{\circled{4}}{=}
    \X^{-1}\A(\K)^{r+1}$, where $\circled{3}$ follows from the induction hypothesis and
$\circled{4}$ follows from \eqref{eq:self_adjoint_A}. Hence, $(\A(\K)^r)^\top\X^{-1}
    =
    \X^{-1}\A(\K)^r$ holds for every $r\in\mathbb{Z}_{\geq0}$.    
\end{proof}

\label{subsec:exact_gradient}
Lemma~\ref{lem:self_adjoint_power} allows the terms involving transposed
powers of the transition matrix in the gradient expression $\partial J(\bm k)/\partial k_i$ to be rewritten
using the convolution of the voltage trajectory. 
Specifically, using
\eqref{eq:de_dki} and the fact $\E_i\bm{\tilde{v}}_s
=
\tilde{v}_{s,i}\e_i$, we obtain
\begin{align}
    \frac{\partial}{\partial k_i}
    \left(
        \bm{\tilde{v}}_t^\top
        \X^{-1}
        \bm{\tilde{v}}_t
    \right)
    &\stackrel{\circled{1}}{=}
    2
    \bm{\tilde{v}}_t^\top
    \X^{-1}
    \frac{\partial\bm{\tilde{v}}_t}{\partial k_i}
    \notag\\
    &\stackrel{\circled{2}}{=}
    -2
    \sum_{s=0}^{t-1}
    \tilde{v}_{s,i}
    \bm{\tilde{v}}_t^\top
    \X^{-1}
    \A(\K)^{t-1-s}
    \X\e_i
    \notag\\
    &\stackrel{\circled{3}}{=}
    -2
    \sum_{s=0}^{t-1}
    \tilde{v}_{s,i}
    \left(
        \A(\K)^{t-1-s}
        \bm{\tilde{v}}_t
    \right)_i
    \notag\\
    &\stackrel{\circled{4}}{=}
    -2
    \sum_{s=0}^{t-1}
    \tilde{v}_{s,i}
    \tilde{v}_{2t-1-s,i},
    \label{eq:gain_gradient_step1}
\end{align}
where $\circled{3}$ follows from
$\X^{-1}\A(\K)^{t-1-s}\X
=
(\A(\K)^{t-1-s})^\top$,
which is implied by
\eqref{eq:self_adjoint_power}. For $\circled{4}$, equation
\eqref{eq:trajectory} gives $\A(\K)^{t-1-s}\bm{\tilde{v}}_t
    =
    \A(\K)^{2t-1-s}\bm{\tilde{v}}_0
    =
    \bm{\tilde{v}}_{2t-1-s}$. Finally, summing \eqref{eq:gain_gradient_step1} over the horizon and
including the derivative of the gain regularization term gives
\begin{equation}
\begin{aligned}
    \frac{\partial J(\kvec)}{\partial k_i}
    &=
    \sum_{t=1}^{T}
    \frac{\partial}{\partial k_i}
    \left(
        \bm{\tilde{v}}_t^\top
        \X^{-1}
        \bm{\tilde{v}}_t
    \right)
    +
    {\lambda_i} k_i
    \\
    &=
    -2
    \sum_{t=1}^{T}
    \sum_{s=0}^{t-1}
    \tilde{v}_{s,i}
    \tilde{v}_{2t-1-s,i}
    +
    {\lambda_i} k_i,
\end{aligned}
\end{equation}
which proves~\eqref{eq:local_regularized_gradient_main}.

\subsection{Projected Gradient Descent}

Let $\bm\alpha:=(\alpha_1,\ldots,\alpha_N)^\top$ collect the
local step sizes in~\eqref{eq:projected_gain_update}, and set
$\bm H:=\operatorname{Diag}(\bm\alpha)$.
The projected gradient descent update is
$T_{\bm\alpha}(\kvec):=
\proj_{\calD}(\kvec-\bm H\nabla J(\kvec))$,
where $\proj_{\calD}$ denotes the Euclidean projection onto $\calD$.
Since $\calD$ is the Cartesian product of the local gain intervals,
this projection decomposes into independent projections of each
 $k_i$ onto $[k_i^{\min},k_i^{\max}]$.
By Theorem~\ref{thm:exact_local_gradient}, when there are no exogenous disturbances, the individual gain updates
in~\eqref{eq:projected_gain_update} follow this projected gradient
update. The proposed algorithm therefore performs exact projected
gradient steps for~\eqref{eq:general_gain_objective2}.

When exogenous injections vary at each step, the voltage transition
becomes $\tilde{\bm v}_{t+1}=\A(\K)\tilde{\bm v}_t+\bm w_t$,
where $\bm w_t:=\Rmat(\bm p_{t+1}-\bm p_t)$ is the disturbance incurred from time-varying active power injections. For this setting, we first analyze voltage stability under the following assumptions.

\begin{assumption}
\label{ass:bounded_perturbations}
There exists a finite constant $r_p\ge0$ such that $\|\bm w_t^{(m)}\|\le r_p$ for every $m\ge0$ and $0\le t\le2T-2.$
\end{assumption}

\begin{assumption}
\label{ass:admissible_gains}
The admissible
gains satisfy $\bm0\prec\K\prec2\X^{-1}$ for every $\kvec\in\calD$.
\end{assumption}

Let $\lambda_{\max}(\X)$ and $\lambda_{\min}(\X)$ denote the largest
and smallest eigenvalues of $\X$, respectively. The following theorem bounds the voltage deviations under the
gain updates and exogenous perturbations.

\begin{theorem}
\label{thm:voltage_stability}
Under Assumptions~\ref{ass:bounded_perturbations}
and~\ref{ass:admissible_gains}, there exists $\rho\in(0,1)$ such that
the voltage deviations
generated by Algorithm~\ref{alg:local_gain_update} satisfy
\vspace{-0.2cm}
\begin{equation}
\begin{aligned}
\|\bm{\tilde{v}}_t^{(m)}\|
&\le
\sqrt{\frac{\lambda_{\max}(\X)}{\lambda_{\min}(\X)}}
\rho^{(2T-1)m+t}\|\bm{\tilde{v}}_0^{(0)}\|\\
&\quad+
\sqrt{\frac{\lambda_{\max}(\X)}{\lambda_{\min}(\X)}}
\frac{1-\rho^{(2T-1)m+t}}{1-\rho}r_p,
\end{aligned}
\label{eq:voltage_stability_bound}
\end{equation}
for all $m\ge0$, $0\le t\le2T-1$.
For $m\to\infty$, the voltage deviations are bounded by
    \vspace{-0.2cm}
\begin{equation}
\limsup_{m\to\infty}\max_{0\le t\le2T-1}
\|\bm{\tilde{v}}_t^{(m)}\|
\le
\sqrt{\frac{\lambda_{\max}(\X)}{\lambda_{\min}(\X)}}
\frac{r_p}{1-\rho}.
\label{eq:voltage_stability_limiting_bound}
\end{equation}
\end{theorem}

The proof is included in Appendix A of~\cite{dong2026decentralizedgainlearningvoltage}.

\begin{figure}[t]
  \centering

  \begin{minipage}[t]{0.499\linewidth}
    \centering
    \includegraphics[width=\linewidth]{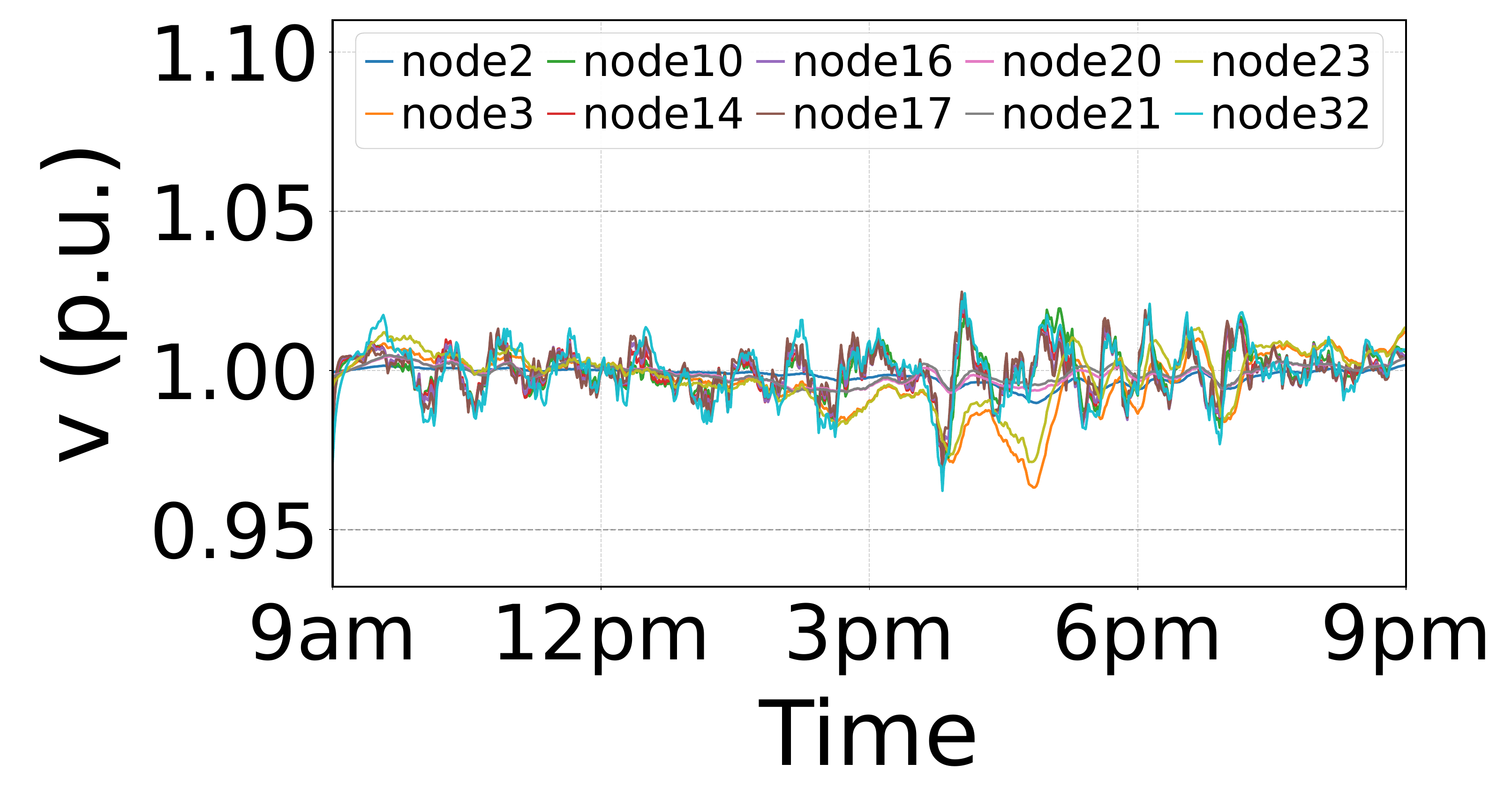}
  \end{minipage}%
  \hspace{0.002\linewidth}%
  \begin{minipage}[t]{0.499\linewidth}
    \centering
    \includegraphics[width=\linewidth]{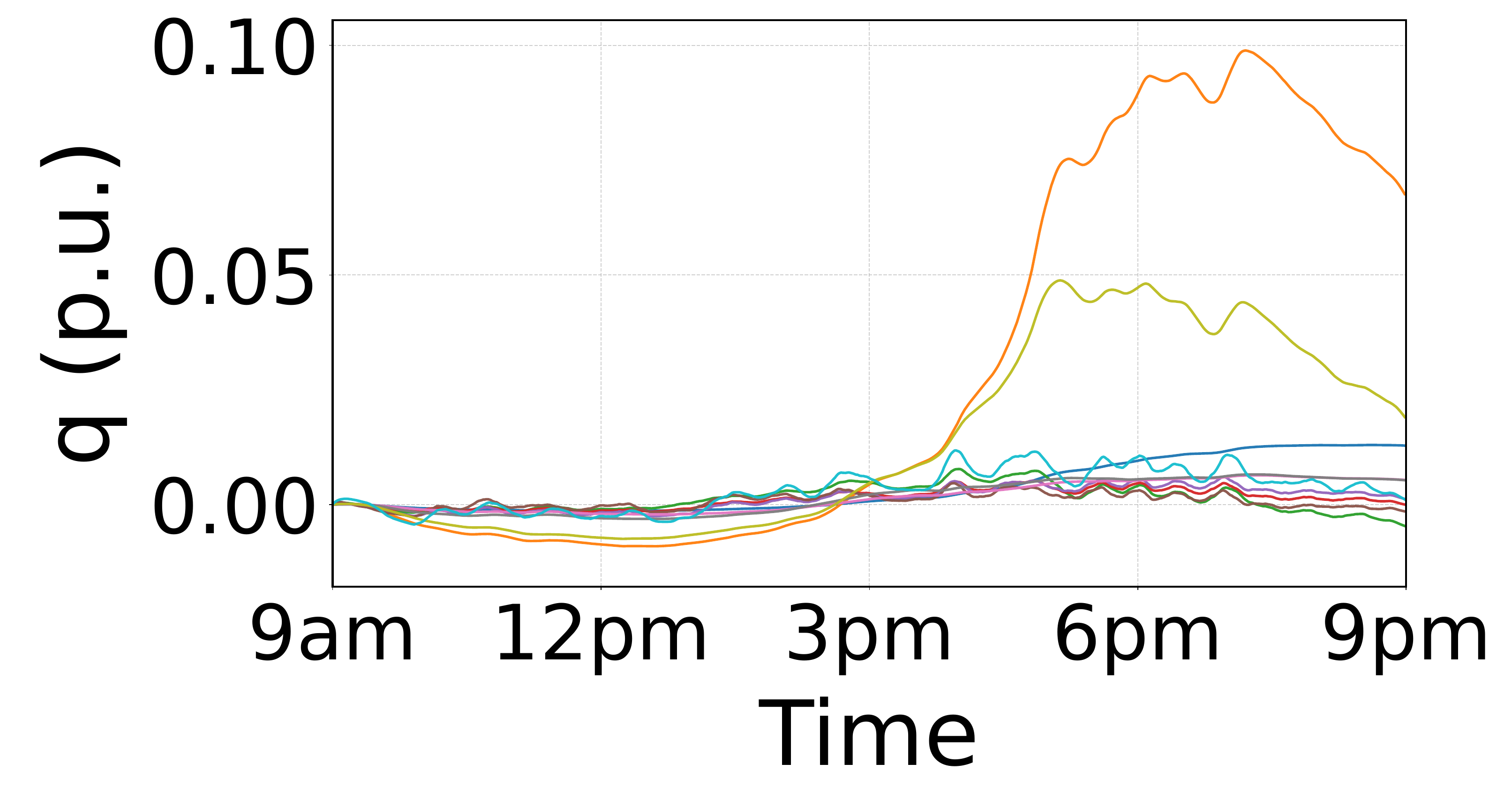}
  \end{minipage}\\[-0.5em]
  {\footnotesize (a) Proposed controller with decentralized gain update.}\\[0.7em]

  \begin{minipage}[t]{0.499\linewidth}
    \centering
    \includegraphics[width=\linewidth]{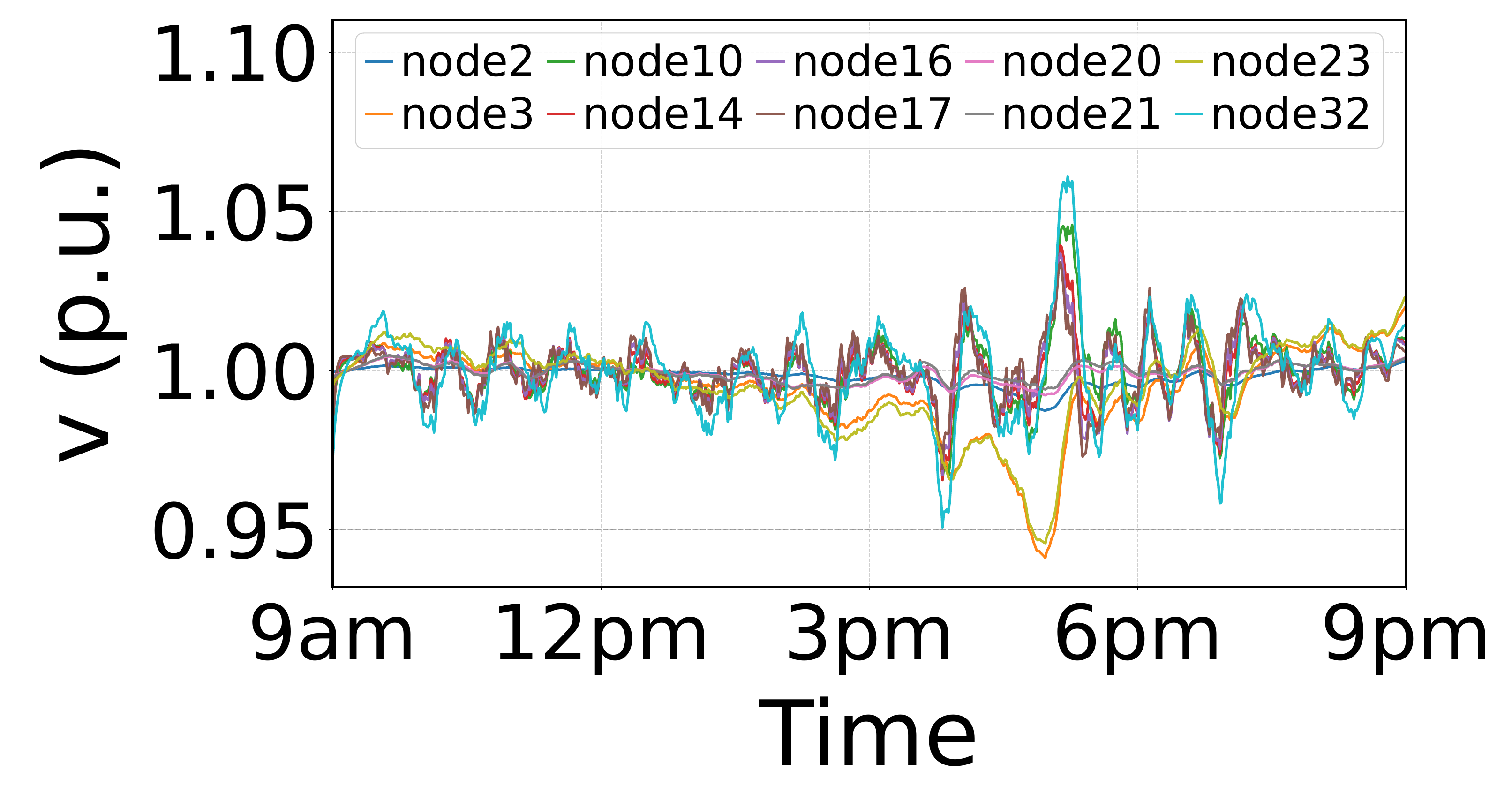}
  \end{minipage}%
  \hspace{0.002\linewidth}%
  \begin{minipage}[t]{0.499\linewidth}
    \centering
    \includegraphics[width=\linewidth]{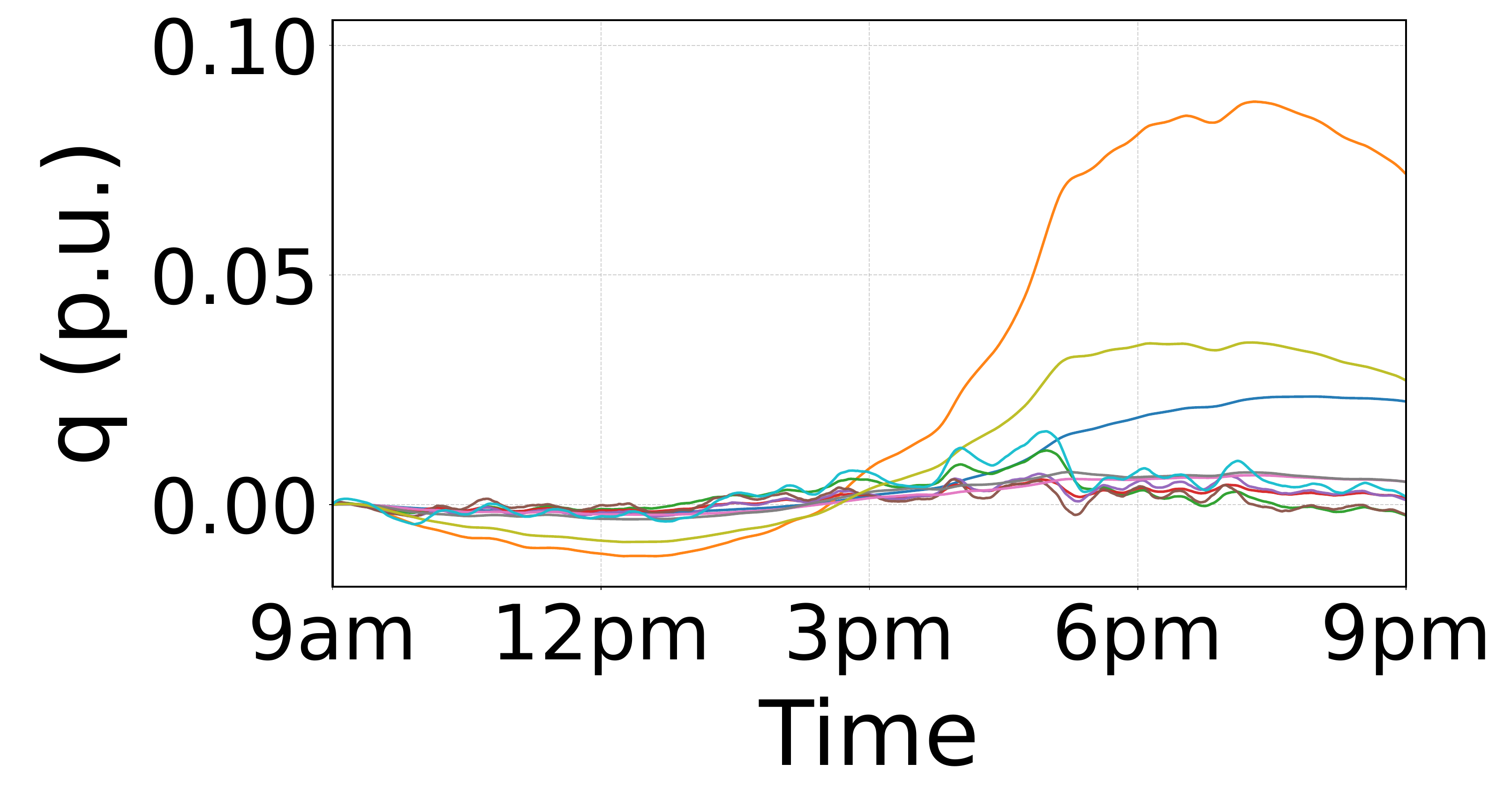}
  \end{minipage}\\[-0.5em]
  {\footnotesize (b) Fixed-gain linear controller.}
  \vspace{-0.1cm}
  \caption{Voltage trajectories (left) and reactive power actions (right) at 10 buses under the proposed controller and the fixed-gain linear controller.}
  \label{fig:varying_voltage_control}
 \vspace{-1em}
\end{figure}

Under Assumptions~\ref{ass:bounded_perturbations}--\ref{ass:admissible_gains}, Lemma~3 of the online version of this paper~\cite{dong2026decentralizedgainlearningvoltage} shows that the gradient of the objective admits a uniform Lipschitz constant
$L>0$ on $\calD$. We impose the following additional
condition on step sizes for convergence analysis of our proposed gain update algorithm.


\begin{assumption}
\label{ass:stepsizes}
The step sizes satisfy $\max\limits_{i\in\calN} \alpha_i < 2/L$.
\end{assumption}

The following theorem gives
average squared gain increments and show that consecutive gain increments vanish in the disturbance-free case.

\vspace{-0.2cm}
\begin{theorem}
\label{thm:descent_stationarity}
Let Assumptions~\ref{ass:bounded_perturbations}--\ref{ass:stepsizes} hold. There exist finite constants
$C_{\bm\alpha,1}>0$ and $C_{\bm\alpha,2}\ge0$ such that
\begin{equation}
\label{eq:average_gain_updates}
\limsup_{M\to\infty}
\frac1M
\sum_{m=0}^{M-1}
\|\kvec^{(m+1)}-\kvec^{(m)}\|^2
\le
C_{\bm\alpha,1}r_p^2
+
C_{\bm\alpha,2}r_p^4.
\end{equation}
If $r_p=0$, then $\|\kvec^{(m+1)}-\kvec^{(m)}\|\to0$.
\end{theorem}

We provide the proof in Appendix B of~\cite{dong2026decentralizedgainlearningvoltage}.

\section{Case Study}
\label{sec:case_study}

\begin{figure}[t]
  \centering
  \includegraphics[width=0.9\linewidth]{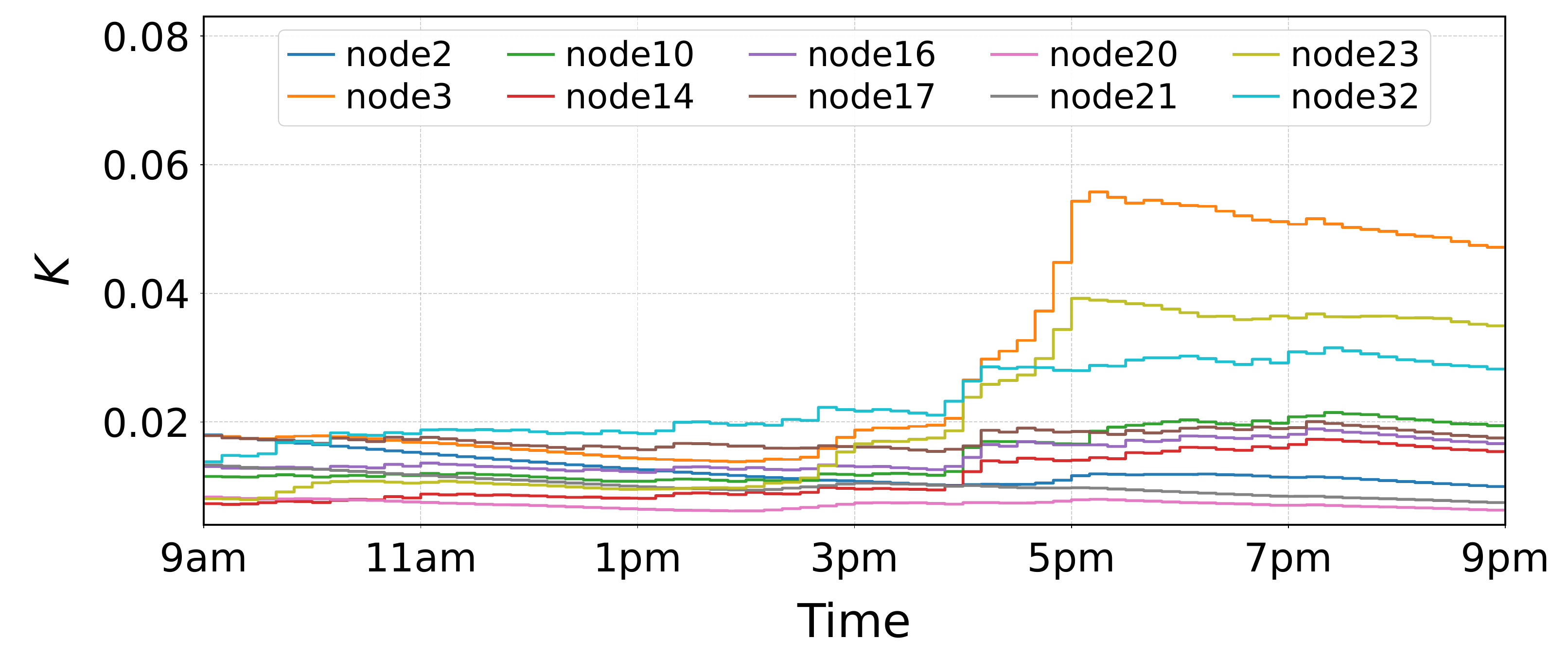}
    \vspace{-0.8em}
  \caption{Gain trajectories at 10 buses under the proposed gain update method.}

\label{fig:varying_gain} 
\vspace{-0.2em}
\end{figure}

\begin{table}[t]
    \centering
    \caption{Voltage regulation performance across 30 days.}
    \label{tab:30day_controller_comparison}
    \begin{tabular}{lcc}
        \toprule
        Metric
        & \shortstack{Fixed-gain\\linear control}
        & \shortstack{Proposed\\method} \\
        \midrule
        Mean daily cumulative cost
        & 1.035 & \textbf{0.614} \\
        Mean daily violation duration (min)
        & 20.77 & \textbf{0.90} \\
        \bottomrule
    \end{tabular}
    \vspace{-1em}
\end{table}

We evaluate the proposed method on the IEEE
33-bus test feeder~\cite{baran1989network}. All electrical
quantities are expressed in per unit on a 100~MVA, 12.66~kV base, and the
voltage reference is set to $1$~p.u. Although the controller is designed under
the LinDistFlow model, the case study is simulated using the actual nonlinear
DistFlow equations~\cite{chiang2002existence}. The proposed controller with decentralized gain update is compared with a fixed-gain linear controller that uses the same decentralized structure in \eqref{eq:incremental_controller}. The only difference is that the proposed controller updates the feedback gains online, whereas the baseline keeps the gains fixed. The code for this work is publicly available at \url{https://github.com/nyudyw/Decentralized-Gain-Update}.

To evaluate the voltage regulation performance of the proposed method under changing
operating conditions, we consider time-varying load and generation profiles
constructed from the Caltech SoCal distribution grid digital twin
dataset~\cite{10.1145/3744255.3811717}. The dataset was collected from an
operational distribution system and captures heterogeneous
generation and demand associated with solar photovoltaic generation, data
centers, electric vehicle charging, and cooling facilities. In the
experiment, voltage measurements are collected and reactive power commands are
updated every minute.

Fig.~\ref{fig:varying_voltage_control} compares the closed-loop behaviors of
the two controllers. The proposed controller with decentralized gain update generally maintains bus voltages
close to $1$~p.u., as shown in Fig.~\ref{fig:varying_voltage_control} (a). When the load increases substantially after 3 pm, the operating
condition departs from the regime for which the initial gains are well suited.
The proposed method responds by increasing local feedback gains, as
shown in Fig.~\ref{fig:varying_gain}, and rapidly adjusting the reactive power
outputs. As a result, all bus voltages remain within the
$0.95$--$1.05$~p.u. limits throughout the simulated period. In contrast, the
fixed-gain linear controller retains the same gain and responds less effectively to changes in power injections, resulting
in larger voltage deviations and intermittent voltage violations, as in Fig.~\ref{fig:varying_voltage_control} (b).
On average, the proposed controller with gain update reduces the voltage deviation by approximately
$30.93\%$ relative to the fixed-gain linear controller, while increasing
the reactive power magnitude by only $8.47\%$.

We further test the performance of two controllers using power profiles from
30 different days. The result is summarized in Table~\ref{tab:30day_controller_comparison}.
 For mean daily cumulative cost, we sum the costs over each day’s simulation period using the same system-level cost structure as in \eqref{eq:regularized_objective_main} and then average the daily totals across 30 days. The proposed controller with decentralized gain update achieves a 40.66\% lower mean daily
cost than the fixed-gain linear controller.
It also reduces the average daily voltage violation (voltage outside $[0.95,1.05]$~p.u.) duration
from 20.77 to 0.90~minutes, corresponding to a 95.67\% reduction. These results demonstrate that the proposed controller with gain update retains the decentralized
implementation advantages of local control while improving its ability to
rapidly adapt to changing operating conditions.

\section{Conclusion}
\label{sec:conclusion}

This paper develops a decentralized gain update algorithm for local
Volt/VAR control. By exploiting the structure of distribution flow equations, the gradient of a system-level voltage regulation objective with respect to the feedback gain is decomposed into locally computable components, enabling projected gradient descent without access to network models or communication among buses. The proposed controller remains fully decentralized while adapting online to
changing operating conditions. Numerical studies demonstrate its effectiveness for voltage regulation under time
varying load and photovoltaic generation.

Future work will investigate how to incorporate alternative measures of voltage deviation and economic
objectives, and establish convergence to the global objective functions. For these objectives, we will
examine whether their structure and the system dynamics allow
update directions to be computed only using local  information.
Further extensions include optimizing active and reactive power
setpoints, incorporating hard constraints on
bus voltages, and considering
nonlinear system models.


\bibliographystyle{IEEEtran}
\bibliography{references}

\newpage

\appendix
\subsection{Proof of Theorem~\ref{thm:voltage_stability}}
\label{app:voltage_stability}

To prove voltage stability, we first analyze how voltage evolves as algorithm proceeds. The following lemma establishes a common contraction bound for all
admissible gains. We then apply it within and across rounds to
bound the voltage deviations and accumulated perturbations.

\begin{lemma}
\label{lem:app:stability_contraction}
Under Assumption~\ref{ass:admissible_gains}, there exists
$\rho\in(0,1)$ such that
\begin{equation}
\bigl\|\X^{-1/2}\A(\K)\X^{1/2}\bigr\|_2\le\rho,
\qquad \forall\kvec\in\calD.
\label{eq:app:stability_contraction}
\end{equation}
\end{lemma}

\begin{proof}
Since $\A(\K)=\Id-\X\K$, its similarity transform satisfies
\[
\X^{-1/2}\A(\K)\X^{1/2}
=\Id-\X^{1/2}\K\X^{1/2}.
\]
Assumption~\ref{ass:admissible_gains} gives
$\bm0\prec\K\prec2\X^{-1}$. Multiplying these inequalities
on both sides by $\X^{1/2}$ yields
\[
\bm0\prec\X^{1/2}\K\X^{1/2}\prec2\Id.
\]
Therefore,
\[
-\Id\prec\Id-\X^{1/2}\K\X^{1/2}\prec\Id.
\]
The transformed matrix is symmetric, with all eigenvalues strictly
between $-1$ and $1$, so its spectral norm is strictly less than
one. This norm is continuous in $\kvec$ and attains its maximum
on the compact gain box $\calD$. Since the gain condition holds
throughout $\calD$, this maximum is strictly less than one.
Hence, there exists $\rho\in(0,1)$ satisfying
\eqref{eq:app:stability_contraction}.
\end{proof}

We now apply this contraction bound to the voltage dynamics.
For each round, define
\begin{equation}
\rho_m:=
\bigl\|\X^{-1/2}\A(\K^{(m)})\X^{1/2}\bigr\|_2,
\label{eq:app:round_contraction}
\end{equation}
where $\rho_m$ is the contraction factor for the gains applied
in round $m$. Lemma~\ref{lem:app:stability_contraction} therefore gives
$0\le\rho_m\le\rho<1$ for every round.

We next use the factor $\rho_m$ to bound the voltage deviations
within round $m$. The gains remain fixed during the round, and
the voltage dynamics satisfy
\[
\bm{\tilde{v}}_{t+1}^{(m)}
=\A(\K^{(m)})\bm{\tilde{v}}_t^{(m)}+\bm w_t^{(m)}.
\]
The disturbance bound gives
\[
\|\X^{-1/2}\bm w_t^{(m)}\|
\le\|\X^{-1/2}\|_2\|\bm w_t^{(m)}\|
\le\frac{r_p}{\sqrt{\lambda_{\min}(\X)}}.
\]
Using~\eqref{eq:app:round_contraction} and the triangle inequality
therefore yields
\[
\bigl\|\X^{-1/2}\bm{\tilde{v}}_{t+1}^{(m)}\bigr\|
\le\rho_m\bigl\|\X^{-1/2}\bm{\tilde{v}}_t^{(m)}\bigr\|
+\frac{r_p}{\sqrt{\lambda_{\min}(\X)}}.
\]
Repeated substitution within the round gives
\begin{equation}
\begin{aligned}
\bigl\|\X^{-1/2}\bm{\tilde{v}}_t^{(m)}\bigr\|
&\le\rho_m^t
\bigl\|\X^{-1/2}\bm{\tilde{v}}_0^{(m)}\bigr\|\\
&\quad+\sum_{s=0}^{t-1}\rho_m^{t-1-s}
\bigl\|\X^{-1/2}\bm w_s^{(m)}\bigr\|\\
&\le\rho_m^t
\bigl\|\X^{-1/2}\bm{\tilde{v}}_0^{(m)}\bigr\|\\
&\quad+\frac{1-\rho_m^t}
{(1-\rho_m)\sqrt{\lambda_{\min}(\X)}}r_p,
\end{aligned}
\label{eq:app:stability_within_round}
\end{equation}
for $0\le t\le2T-1$. Empty sums equal zero, and $\rho_m^0=1$.

The terminal voltage deviation of one round is the initial
deviation of the next. We can therefore apply the preceding
estimate across successive gain updates. In particular,
$\bm{\tilde{v}}_0^{(m+1)}=\bm{\tilde{v}}_{2T-1}^{(m)}$ and
\eqref{eq:app:stability_within_round} give
\begin{equation}
\begin{aligned}
\bigl\|\X^{-1/2}\bm{\tilde{v}}_0^{(m+1)}\bigr\|
&\le\rho_m^{2T-1}
\bigl\|\X^{-1/2}\bm{\tilde{v}}_0^{(m)}\bigr\|\\
&\quad+\frac{1-\rho_m^{2T-1}}
{(1-\rho_m)\sqrt{\lambda_{\min}(\X)}}r_p.
\end{aligned}
\label{eq:app:stability_round_recursion}
\end{equation}
The updated gains may give a different factor in the next round.
Keeping the factor associated with each round, repeated
substitution yields
\begin{equation}
\begin{aligned}
&\bigl\|\X^{-1/2}\bm{\tilde{v}}_0^{(m)}\bigr\|
\le\left(\prod_{\ell=0}^{m-1}\rho_\ell^{2T-1}\right)
\bigl\|\X^{-1/2}\bm{\tilde{v}}_0^{(0)}\bigr\|\\
&\quad+\frac{r_p}{\sqrt{\lambda_{\min}(\X)}}
\sum_{j=0}^{m-1}
\left(\prod_{\ell=j+1}^{m-1}\rho_\ell^{2T-1}\right)
\frac{1-\rho_j^{2T-1}}{1-\rho_j}.
\end{aligned}
\label{eq:app:stability_boundary_bound}
\end{equation}
Here, an empty product equals one.
The first term contains one contraction factor for each completed
round. In the second term, the geometric sum
$(1-\rho_j^{2T-1})/(1-\rho_j)$ accounts for the perturbations
during round $j$. Their contribution is then reduced by the
contraction factors of rounds $j+1,\ldots,m-1$.

Combining the bound at the beginning of round $m$ with the
first $t$ steps of that round gives
\begin{equation}
\begin{aligned}
&\bigl\|\X^{-1/2}\bm{\tilde{v}}_t^{(m)}\bigr\|
\le\rho_m^t
\left(\prod_{\ell=0}^{m-1}\rho_\ell^{2T-1}\right)
\bigl\|\X^{-1/2}\bm{\tilde{v}}_0^{(0)}\bigr\|\\
&\quad+\frac{r_p\rho_m^t}{\sqrt{\lambda_{\min}(\X)}}
\sum_{j=0}^{m-1}
\left(\prod_{\ell=j+1}^{m-1}\rho_\ell^{2T-1}\right)
\frac{1-\rho_j^{2T-1}}{1-\rho_j}\\
&\quad+\frac{1-\rho_m^t}
{(1-\rho_m)\sqrt{\lambda_{\min}(\X)}}r_p.
\end{aligned}
\label{eq:app:stability_weighted_product_bound}
\end{equation}
The second term accounts for perturbations in previous rounds,
while the third accounts for those in the current round.
The factor $\rho_m^t$ represents the contraction during the
first $t$ steps under the current gains.

We now express this estimate in terms of the voltage
deviations. Since $\X\succ0$, every vector $\bm z$ satisfies
\[
\|\bm z\|\le\sqrt{\lambda_{\max}(\X)}\|\X^{-1/2}\bm z\|,
\quad
\|\X^{-1/2}\bm z\|
\le\frac{\|\bm z\|}{\sqrt{\lambda_{\min}(\X)}}.
\]
Applying these inequalities
to~\eqref{eq:app:stability_weighted_product_bound} gives
\begin{equation}
\begin{aligned}
&\|\bm{\tilde{v}}_t^{(m)}\|\le
\sqrt{\frac{\lambda_{\max}(\X)}{\lambda_{\min}(\X)}}
\rho_m^t\left(\prod_{\ell=0}^{m-1}\rho_\ell^{2T-1}\right)
\|\bm{\tilde{v}}_0^{(0)}\|\\
&\quad+
\sqrt{\frac{\lambda_{\max}(\X)}{\lambda_{\min}(\X)}}\,r_p\rho_m^t
\sum_{j=0}^{m-1}
\left(\prod_{\ell=j+1}^{m-1}\rho_\ell^{2T-1}\right)
\frac{1-\rho_j^{2T-1}}{1-\rho_j}\\
&\quad+
\sqrt{\frac{\lambda_{\max}(\X)}{\lambda_{\min}(\X)}}
\frac{1-\rho_m^t}{1-\rho_m}r_p.
\end{aligned}
\label{eq:app:round_dependent_voltage_bound}
\end{equation}
This bound accounts for the gains applied in each
round through the product of their contraction factors.

To obtain a bound that holds uniformly over all admissible gain
sequences, we use $\rho_m\le\rho$. The coefficient of the initial
voltage deviation satisfies
\[
\rho_m^t\prod_{\ell=0}^{m-1}\rho_\ell^{2T-1}
\le\rho^{(2T-1)m+t}.
\]
For the perturbation terms, the geometric sums satisfy
\[
\begin{aligned}
\frac{1-\rho_j^{2T-1}}{1-\rho_j}
&=\sum_{s=0}^{2T-2}\rho_j^s
\le\frac{1-\rho^{2T-1}}{1-\rho},\\
\frac{1-\rho_m^t}{1-\rho_m}
&=\sum_{s=0}^{t-1}\rho_m^s
\le\frac{1-\rho^t}{1-\rho}.
\end{aligned}
\]
Consequently,
\[
\begin{aligned}
&\rho_m^t\sum_{j=0}^{m-1}
\left(\prod_{\ell=j+1}^{m-1}\rho_\ell^{2T-1}\right)
\frac{1-\rho_j^{2T-1}}{1-\rho_j}
+\frac{1-\rho_m^t}{1-\rho_m}\\
&\quad\le
\rho^t\frac{1-\rho^{2T-1}}{1-\rho}
\sum_{j=0}^{m-1}\rho^{(2T-1)(m-1-j)}
+\frac{1-\rho^t}{1-\rho}\\
&\quad=
\rho^t\frac{1-\rho^{(2T-1)m}}{1-\rho}
+\frac{1-\rho^t}{1-\rho}
=\frac{1-\rho^{(2T-1)m+t}}{1-\rho}.
\end{aligned}
\]
Substituting these estimates
into~\eqref{eq:app:stability_weighted_product_bound} gives
\begin{equation}
\begin{aligned}
\bigl\|\X^{-1/2}\bm{\tilde{v}}_t^{(m)}\bigr\|
&\le\rho^{(2T-1)m+t}
\bigl\|\X^{-1/2}\bm{\tilde{v}}_0^{(0)}\bigr\|\\
&\quad+\frac{1-\rho^{(2T-1)m+t}}
{(1-\rho)\sqrt{\lambda_{\min}(\X)}}r_p.
\end{aligned}
\label{eq:app:voltage_bound}
\end{equation}
Applying the norm comparisons above yields
\[
\begin{aligned}
\|\bm{\tilde{v}}_t^{(m)}\|
&\le
\sqrt{\frac{\lambda_{\max}(\X)}{\lambda_{\min}(\X)}}
\rho^{(2T-1)m+t}\|\bm{\tilde{v}}_0^{(0)}\|\\
&\quad+
\sqrt{\frac{\lambda_{\max}(\X)}{\lambda_{\min}(\X)}}
\frac{1-\rho^{(2T-1)m+t}}{1-\rho}r_p,
\end{aligned}
\]
which proves~\eqref{eq:voltage_stability_bound}.

It remains to establish uniform boundedness and the limiting
bound. The right-hand side of~\eqref{eq:voltage_stability_bound}
is a convex combination of
$\sqrt{\lambda_{\max}(\X)/\lambda_{\min}(\X)}
\|\bm{\tilde{v}}_0^{(0)}\|$
and
$\sqrt{\lambda_{\max}(\X)/\lambda_{\min}(\X)}
r_p/(1-\rho)$,
with weights $\rho^{(2T-1)m+t}$ and
$1-\rho^{(2T-1)m+t}$. Therefore,
\[
\|\bm{\tilde{v}}_t^{(m)}\|
\le
\sqrt{\frac{\lambda_{\max}(\X)}{\lambda_{\min}(\X)}}
\max\left\{
\|\bm{\tilde{v}}_0^{(0)}\|,
\frac{r_p}{1-\rho}
\right\},
\]
for all $m\ge0$ and $0\le t\le2T-1$.
Thus, the voltage deviations remain uniformly bounded throughout
the gain updates.

Finally, taking the maximum over $0\le t\le2T-1$ in
\eqref{eq:voltage_stability_bound}, and using
$\rho^{(2T-1)m+t}\le\rho^{(2T-1)m}$ and
$1-\rho^{(2T-1)m+t}\le1$, gives
\[
\begin{aligned}
\max_{0\le t\le2T-1}\|\bm{\tilde{v}}_t^{(m)}\|
&\le
\sqrt{\frac{\lambda_{\max}(\X)}{\lambda_{\min}(\X)}}
\rho^{(2T-1)m}\|\bm{\tilde{v}}_0^{(0)}\|\\
&\quad+
\sqrt{\frac{\lambda_{\max}(\X)}{\lambda_{\min}(\X)}}
\frac{r_p}{1-\rho}.
\end{aligned}
\]
The first term tends to zero as $m\to\infty$. Taking the limit
superior proves~\eqref{eq:voltage_stability_limiting_bound}.

\subsection{Proof of Theorem~\ref{thm:descent_stationarity}}
\label{app:descent_stationarity}

The proof uses the voltage bound from
Theorem~\ref{thm:voltage_stability} and proceeds through three lemmas.
We first obtain a common Lipschitz constant, then bound the local
direction error and derive a descent inequality. Finally, we sum
these inequalities to bound the gain changes across rounds.

For each round $m$, let $J^{(m)}(\kvec)$ denote the objective
in~\eqref{eq:regularized_objective_main} evaluated using that
round's initial voltage deviation and realized perturbations.
These data are held fixed when the candidate gain $\kvec$ is varied.

\begin{lemma}
\label{lem:app:uniform_bounds}
Under Assumptions~\ref{ass:bounded_perturbations}--\ref{ass:admissible_gains},
the voltage trajectories are uniformly bounded over all
rounds and admissible gains, and the gradients $\nabla J^{(m)}$
admit a common finite Lipschitz constant $L>0$
on $\calD$ for all rounds.
\end{lemma}

\begin{proof}
We use Lemma~\ref{lem:app:stability_contraction} and the weighted
voltage bound~\eqref{eq:app:voltage_bound}, with the same
$\rho\in(0,1)$.
Define the uniform bound
\begin{equation}
B:=\max\left\{
\bigl\|\X^{-1/2}\bm{\tilde{v}}_0^{(0)}\bigr\|,
\frac{r_p}{(1-\rho)\sqrt{\lambda_{\min}(\X)}}
\right\}.
\label{eq:app:uniform_voltage_bound}
\end{equation}
The right-hand side of~\eqref{eq:app:voltage_bound} is a convex
combination of the two quantities defining $B$. Therefore,
$\|\X^{-1/2}\bm{\tilde{v}}_t^{(m)}\|\le B$.

For a fixed round $m$ and any candidate gain $\kvec\in\calD$,
write $\bm{\tilde{v}}_t$ for the trajectory used to evaluate
$J^{(m)}(\kvec)$. With the initial deviation $\bm{\tilde{v}}_0^{(m)}$
and perturbations of round $m$ held fixed, this trajectory is
\begin{equation}
\bm{\tilde{v}}_t
=\A(\K)^t\bm{\tilde{v}}_0^{(m)}
+\sum_{j=0}^{t-1}\A(\K)^{t-1-j}\bm w_j^{(m)}.
\label{eq:app:forced_trajectory}
\end{equation}
The corresponding objective is
\begin{equation}
J^{(m)}(\kvec)
=\sum_{t=1}^{T}
\bm{\tilde{v}}_t^\top\X^{-1}\bm{\tilde{v}}_t
+\kvec^\top\bm\Lambda\kvec.
\label{eq:app:round_objective}
\end{equation}
At $\kvec=\kvec^{(m)}$, the trajectory is the measured
$\bm{\tilde{v}}_t^{(m)}$.
Applying~\eqref{eq:app:stability_contraction}
to~\eqref{eq:app:forced_trajectory} yields
\[
\bigl\|\X^{-1/2}\bm{\tilde{v}}_t\bigr\|
\le\rho^tB
+\frac{1-\rho^t}{(1-\rho)\sqrt{\lambda_{\min}(\X)}}r_p
\le B.
\]
Thus, the same bound holds for every round, every admissible gain,
and $0\le t\le2T-1$. The voltage deviations are also uniformly
bounded in the Euclidean norm, since
$\|\bm{\tilde{v}}_t\|\le\sqrt{\|\X\|_2}B$.

We next bound the Hessian of $J^{(m)}$ in~\eqref{eq:app:round_objective}.
Fix a unit direction
$\bm h\in\R^N$, and let $\partial_{\bm h}$ denote differentiation
along $\kvec+s\bm h$ at $s=0$. The initial deviation and
disturbances remain fixed during differentiation. The trajectory
recursion gives
\[
\begin{aligned}
\partial_{\bm h}\bm{\tilde{v}}_{t+1}
&=\A(\K)\partial_{\bm h}\bm{\tilde{v}}_t
-\X\Diag(\bm h)\bm{\tilde{v}}_t,\\
\partial_{\bm h}^2\bm{\tilde{v}}_{t+1}
&=\A(\K)\partial_{\bm h}^2\bm{\tilde{v}}_t
-2\X\Diag(\bm h)\partial_{\bm h}\bm{\tilde{v}}_t,
\end{aligned}
\]
with both derivatives equal to zero at $t=0$. Since $\|\bm h\|=1$,
\[
\bigl\|\X^{1/2}\Diag(\bm h)\X^{1/2}\bigr\|_2
\le\|\X\|_2\max_i|h_i|\le\|\X\|_2.
\]
Together with the contraction and voltage bounds, this gives
\[
\begin{aligned}
\bigl\|\X^{-1/2}\partial_{\bm h}\bm{\tilde{v}}_{t+1}\bigr\|
&\le\rho\bigl\|\X^{-1/2}\partial_{\bm h}\bm{\tilde{v}}_t\bigr\|
+\|\X\|_2B,\\
\bigl\|\X^{-1/2}\partial_{\bm h}^2\bm{\tilde{v}}_{t+1}\bigr\|
&\le\rho\bigl\|\X^{-1/2}\partial_{\bm h}^2\bm{\tilde{v}}_t\bigr\|\\
&\quad+2\|\X\|_2
\bigl\|\X^{-1/2}\partial_{\bm h}\bm{\tilde{v}}_t\bigr\|.
\end{aligned}
\]
Starting from the zero derivatives and using $\rho\le1$,
induction gives
\begin{equation}
\begin{aligned}
\bigl\|\X^{-1/2}\partial_{\bm h}\bm{\tilde{v}}_t\bigr\|
&\le t\|\X\|_2B,\\
\bigl\|\X^{-1/2}\partial_{\bm h}^2\bm{\tilde{v}}_t\bigr\|
&\le t(t-1)\|\X\|_2^2B.
\end{aligned}
\label{eq:app:trajectory_derivative_bounds}
\end{equation}

Differentiating the quadratic terms in
\eqref{eq:app:round_objective} twice gives
\[
\begin{aligned}
\bm h^\top\nabla^2J^{(m)}(\kvec)\bm h
&=\sum_{i=1}^{N}\lambda_i h_i^2\\
&\quad+2\sum_{t=1}^{T}
(\partial_{\bm h}\bm{\tilde{v}}_t)^\top
\X^{-1}\partial_{\bm h}\bm{\tilde{v}}_t\\
&\quad+2\sum_{t=1}^{T}
\bm{\tilde{v}}_t^\top\X^{-1}
\partial_{\bm h}^2\bm{\tilde{v}}_t.
\end{aligned}
\]
Applying the preceding bounds and the Cauchy--Schwarz inequality
yields
\[
\begin{aligned}
\bigl|\bm h^\top\nabla^2J^{(m)}(\kvec)\bm h\bigr|
&\le\max_i\lambda_i\\
&\quad+2\|\X\|_2^2B^2
\sum_{t=1}^{T}\bigl(t^2+t(t-1)\bigr)\\
&=\max_i\lambda_i
+\frac{T(T+1)(4T-1)}{3}\|\X\|_2^2B^2.
\end{aligned}
\]
Consequently, a common Hessian bound is
\begin{equation}
L
:=\max_i\lambda_i
+\frac{T(T+1)(4T-1)}{3}\|\X\|_2^2B^2.
\label{eq:app:explicit_lipschitz}
\end{equation}
Since the Hessian is symmetric, taking the supremum over unit
directions gives $\|\nabla^2J^{(m)}(\kvec)\|_2\le L$
for every $\kvec\in\calD$ and every round. Convexity of $\calD$
then implies
\[
\|\nabla J^{(m)}(\kvec')-\nabla J^{(m)}(\kvec)\|
\le L\|\kvec'-\kvec\|,
\quad\forall\kvec,\kvec'\in\calD.
\]
Thus, the same Lipschitz constant applies throughout the updates.
\end{proof}

This bound allows the step size condition to hold in every round.
We next compare the direction in~\eqref{eq:local_update_signal}
with the gradient of $J^{(m)}$ for that round. Define
\begin{equation}
e_i^{(m)}
:=d_i^{(m)}-\frac{\partial J^{(m)}}{\partial k_i}(\kvec^{(m)}),
\qquad i\in\calN,
\label{eq:app:direction_error}
\end{equation}
where $\bm e^{(m)}$ collects these direction errors.
The following lemma bounds their effect on the gain update.

\begin{lemma}
\label{lem:app:direction_descent}
Under Assumptions~\ref{ass:bounded_perturbations}--\ref{ass:stepsizes},
there exist finite constants $C_e\ge0$
and $c_{\bm\alpha}>0$ such that
$\|\bm e^{(m)}\|\le C_eBr_p$ and
\begin{equation}
\begin{aligned}
J^{(m)}(\kvec^{(m+1)})
&\le J^{(m)}(\kvec^{(m)})
-c_{\bm\alpha}\|\kvec^{(m+1)}-\kvec^{(m)}\|^2\\
&\quad+\frac{\|\bm e^{(m)}\|^2}{4c_{\bm\alpha}},
\end{aligned}
\label{eq:app:gain_descent}
\end{equation}
for every round $m$.
\end{lemma}

\begin{proof}
Fix a round $m$ and use the candidate trajectory
in~\eqref{eq:app:forced_trajectory}. Differentiating its recursion
with respect to $k_i$ gives
\[
\frac{\partial\bm{\tilde{v}}_t}{\partial k_i}
=-\sum_{s=0}^{t-1}
\A(\K)^{t-1-s}\X\E_i\bm{\tilde{v}}_s.
\]
The initial derivative is zero because $\bm{\tilde{v}}_0^{(m)}$
is held fixed. Applying the chain rule as in~\eqref{eq:gradJ}
and using Lemma~\ref{lem:self_adjoint_power} gives
\begin{equation}
\begin{aligned}
\frac{\partial J^{(m)}}{\partial k_i}(\kvec^{(m)})
&=-2\sum_{t=1}^{T}\sum_{s=0}^{t-1}
\tilde{v}_{s,i}^{(m)}
\bigl[\A(\K^{(m)})^{t-1-s}\bm{\tilde{v}}_t^{(m)}\bigr]_i\\
&\quad+\lambda_i k_i^{(m)}.
\end{aligned}
\label{eq:app:exact_gradient}
\end{equation}
The disturbances between steps $t$ and $2t-1-s$ give
\[
\begin{aligned}
\bm{\tilde{v}}_{2t-1-s}^{(m)}
&=\A(\K^{(m)})^{t-1-s}\bm{\tilde{v}}_t^{(m)}\\
&\quad+\sum_{j=t}^{2t-2-s}
\A(\K^{(m)})^{2t-2-s-j}\bm w_j^{(m)}.
\end{aligned}
\]
The last sum is zero when $s=t-1$. Subtracting
\eqref{eq:app:exact_gradient} from~\eqref{eq:local_update_signal}
therefore yields
\begin{equation}
e_i^{(m)}
=-2\sum_{t=1}^{T}\sum_{s=0}^{t-2}
\tilde{v}_{s,i}^{(m)}
\biggl[\sum_{j=t}^{2t-2-s}
\A(\K^{(m)})^{2t-2-s-j}\bm w_j^{(m)}\biggr]_i.
\label{eq:app:error_expression}
\end{equation}
By~\eqref{eq:app:stability_contraction},
\[
\begin{aligned}
&\left\|\X^{-1/2}\sum_{j=t}^{2t-2-s}
\A(\K^{(m)})^{2t-2-s-j}\bm w_j^{(m)}\right\|\\
&\quad\le\frac{r_p}{\sqrt{\lambda_{\min}(\X)}}
\sum_{j=t}^{2t-2-s}\rho^{2t-2-s-j}
\le\frac{(t-1-s)r_p}{\sqrt{\lambda_{\min}(\X)}}.
\end{aligned}
\]
For every vector $\bm z$,
$\|\bm z\|\le\sqrt{\|\X\|_2}\|\X^{-1/2}\bm z\|$.
The Euclidean norm of a componentwise product is at most the
product of the two Euclidean norms. It follows that
\begin{equation}
\begin{aligned}
\|\bm e^{(m)}\|
&\le\frac{2\|\X\|_2r_p}{\sqrt{\lambda_{\min}(\X)}}
\sum_{t=1}^{T}\sum_{s=0}^{t-2}(t-1-s)
\bigl\|\X^{-1/2}\bm{\tilde{v}}_s^{(m)}\bigr\|\\
&\le\frac{2\|\X\|_2Br_p}{\sqrt{\lambda_{\min}(\X)}}
\sum_{t=1}^{T}\sum_{s=0}^{t-2}(t-1-s).
\end{aligned}
\label{eq:app:error_trajectory_bound}
\end{equation}
Since
\[
2\sum_{t=1}^{T}\sum_{s=0}^{t-2}(t-1-s)
=\sum_{t=1}^{T}t(t-1)
=\frac{T(T+1)(T-1)}{3},
\]
we obtain $\|\bm e^{(m)}\|\le C_eBr_p$, where
\begin{equation}
C_e:=\frac{T(T+1)(T-1)}{3}
\frac{\|\X\|_2}{\sqrt{\lambda_{\min}(\X)}}.
\label{eq:app:error_constant}
\end{equation}

We next derive the descent inequality. The optimality condition
of the scalar projection in~\eqref{eq:projected_gain_update} gives
\[
\bigl(k_i^{(m)}-\alpha_i d_i^{(m)}-k_i^{(m+1)}\bigr)
\bigl(k_i^{(m)}-k_i^{(m+1)}\bigr)\le0.
\]
Dividing by $\alpha_i>0$ and summing over buses gives
\[
\sum_{i=1}^{N}d_i^{(m)}(k_i^{(m+1)}-k_i^{(m)})
\le-\sum_{i=1}^{N}\frac{(k_i^{(m+1)}-k_i^{(m)})^2}{\alpha_i}.
\]
Using the Lipschitz bound and~\eqref{eq:app:direction_error},
\[
\begin{aligned}
J^{(m)}(\kvec^{(m+1)})
&\le J^{(m)}(\kvec^{(m)})\\
&\quad+\nabla J^{(m)}(\kvec^{(m)})^\top
(\kvec^{(m+1)}-\kvec^{(m)})\\
&\quad+\frac{L}{2}
\|\kvec^{(m+1)}-\kvec^{(m)}\|^2\\
&\le J^{(m)}(\kvec^{(m)})\\
&\quad-\sum_{i=1}^{N}
\left(\frac1{\alpha_i}-\frac{L}{2}\right)
(k_i^{(m+1)}-k_i^{(m)})^2\\
&\quad-(\bm e^{(m)})^\top(\kvec^{(m+1)}-\kvec^{(m)}).
\end{aligned}
\]
For the fixed local step sizes, define
\begin{equation}
c_{\bm\alpha}
:=\frac1{2\max_i\alpha_i}-\frac{L}{4}>0.
\label{eq:app:descent_constant}
\end{equation}
The step size condition implies
$1/\alpha_i-L/2\ge2c_{\bm\alpha}$ for every bus.
Consequently,
\[
\begin{aligned}
J^{(m)}(\kvec^{(m+1)})
&\le J^{(m)}(\kvec^{(m)})
-2c_{\bm\alpha}\|\kvec^{(m+1)}-\kvec^{(m)}\|^2\\
&\quad+\|\bm e^{(m)}\|\|\kvec^{(m+1)}-\kvec^{(m)}\|\\
&\le J^{(m)}(\kvec^{(m)})
-c_{\bm\alpha}\|\kvec^{(m+1)}-\kvec^{(m)}\|^2\\
&\quad+\frac{\|\bm e^{(m)}\|^2}{4c_{\bm\alpha}},
\end{aligned}
\]
where the last inequality uses
$ab\le c_{\bm\alpha}b^2+a^2/(4c_{\bm\alpha})$.
This proves~\eqref{eq:app:gain_descent}.
\end{proof}

The descent inequality compares the current and updated gains
using the objective of round $m$. In particular,
$J^{(m)}(\kvec^{(m+1)})$ evaluates the updated gains using round
$m$'s initial deviation and perturbations. The next round's actual
cost is $J^{(m+1)}(\kvec^{(m+1)})$, evaluated using that round's
initial deviation and perturbations. The rest of the proof accounts for
this difference and establishes both conclusions of the theorem.

The rest of the proof for Theorem~\ref{thm:descent_stationarity} is given below.
\begin{proof}
The voltage cost in~\eqref{eq:app:round_objective} is nonnegative
for every candidate gain. In particular,
\[
\begin{aligned}
J^{(m)}(\kvec^{(m+1)})
&\ge(\kvec^{(m+1)})^\top\bm\Lambda\kvec^{(m+1)}\\
&=\frac12\sum_{i=1}^{N}\lambda_i(k_i^{(m+1)})^2.
\end{aligned}
\]
At $\kvec^{(m)}$, the voltage cost is evaluated along the measured
trajectory. Substituting into~\eqref{eq:app:gain_descent} gives
\[
\begin{aligned}
&c_{\bm\alpha}\|\kvec^{(m+1)}-\kvec^{(m)}\|^2
\le\sum_{t=1}^{T}
(\bm{\tilde{v}}_t^{(m)})^\top\X^{-1}\bm{\tilde{v}}_t^{(m)}\\
&\quad+\frac12\sum_{i=1}^{N}\lambda_i
\bigl[(k_i^{(m)})^2-(k_i^{(m+1)})^2\bigr]
+\frac{\|\bm e^{(m)}\|^2}{4c_{\bm\alpha}}.
\end{aligned}
\]
Since each $\lambda_i$ is fixed, summing over $m=0,\ldots,M-1$
telescopes the gain regularization terms:
\begin{equation}
\begin{aligned}
&c_{\bm\alpha}\sum_{m=0}^{M-1}
\|\kvec^{(m+1)}-\kvec^{(m)}\|^2
\le\frac12\sum_{i=1}^{N}\lambda_i
\bigl[(k_i^{(0)})^2-(k_i^{(M)})^2\bigr]\\
&\quad+\sum_{m=0}^{M-1}\sum_{t=1}^{T}
(\bm{\tilde{v}}_t^{(m)})^\top\X^{-1}\bm{\tilde{v}}_t^{(m)}
+\frac1{4c_{\bm\alpha}}
\sum_{m=0}^{M-1}\|\bm e^{(m)}\|^2.
\end{aligned}
\label{eq:app:regularization_sum}
\end{equation}

We next bound the voltage costs. Equation~\eqref{eq:app:voltage_bound}
implies
\[
\begin{aligned}
\bigl\|\X^{-1/2}\bm{\tilde{v}}_t^{(m)}\bigr\|
&\le\rho^{(2T-1)m+t}
\bigl\|\X^{-1/2}\bm{\tilde{v}}_0^{(0)}\bigr\|\\
&\quad+\frac{r_p}{(1-\rho)\sqrt{\lambda_{\min}(\X)}}.
\end{aligned}
\]
Using $(a+b)^2\le2a^2+2b^2$ and summing yields
\[
\begin{aligned}
&\sum_{m=0}^{M-1}\sum_{t=1}^{T}
(\bm{\tilde{v}}_t^{(m)})^\top\X^{-1}\bm{\tilde{v}}_t^{(m)}\\
&\quad\le2\left(\sum_{t=1}^{T}\rho^{2t}\right)
\left(\sum_{m=0}^{M-1}\rho^{2(2T-1)m}\right)
(\bm{\tilde{v}}_0^{(0)})^\top\X^{-1}\bm{\tilde{v}}_0^{(0)}\\
&\qquad+\frac{2MT}{\lambda_{\min}(\X)(1-\rho)^2}r_p^2.
\end{aligned}
\]
Applying the geometric series bound
\[
\sum_{m=0}^{M-1}\rho^{2(2T-1)m}
\le\frac1{1-\rho^{2(2T-1)}}
\]
gives
\begin{equation}
\begin{aligned}
&\sum_{m=0}^{M-1}\sum_{t=1}^{T}
(\bm{\tilde{v}}_t^{(m)})^\top\X^{-1}\bm{\tilde{v}}_t^{(m)}\\
&\quad\le\frac{2\sum_{t=1}^{T}\rho^{2t}}
{1-\rho^{2(2T-1)}}
(\bm{\tilde{v}}_0^{(0)})^\top\X^{-1}\bm{\tilde{v}}_0^{(0)}\\
&\qquad+\frac{2MT}{\lambda_{\min}(\X)(1-\rho)^2}r_p^2.
\end{aligned}
\label{eq:app:voltage_cost_sum}
\end{equation}

For every sample $s$ in round $m$,
\eqref{eq:app:voltage_bound} gives
\[
\begin{aligned}
\bigl\|\X^{-1/2}\bm{\tilde{v}}_s^{(m)}\bigr\|
&\le\rho^{(2T-1)m}
\bigl\|\X^{-1/2}\bm{\tilde{v}}_0^{(0)}\bigr\|\\
&\quad+\frac{r_p}{(1-\rho)\sqrt{\lambda_{\min}(\X)}}.
\end{aligned}
\]
Substituting into the first inequality
of~\eqref{eq:app:error_trajectory_bound} yields
\[
\begin{aligned}
\|\bm e^{(m)}\|
&\le C_er_p\rho^{(2T-1)m}
\bigl\|\X^{-1/2}\bm{\tilde{v}}_0^{(0)}\bigr\|\\
&\quad+\frac{C_er_p^2}{(1-\rho)\sqrt{\lambda_{\min}(\X)}}.
\end{aligned}
\]
Squaring, summing, and using the same geometric series bound gives
\begin{equation}
\begin{aligned}
\sum_{m=0}^{M-1}\|\bm e^{(m)}\|^2
&\le\frac{2C_e^2r_p^2}{1-\rho^{2(2T-1)}}
(\bm{\tilde{v}}_0^{(0)})^\top\X^{-1}\bm{\tilde{v}}_0^{(0)}\\
&\quad+\frac{2MC_e^2}{\lambda_{\min}(\X)(1-\rho)^2}r_p^4.
\end{aligned}
\label{eq:app:direction_error_sum}
\end{equation}
Substituting~\eqref{eq:app:voltage_cost_sum}
and~\eqref{eq:app:direction_error_sum}
into~\eqref{eq:app:regularization_sum} and dropping the nonpositive
terminal regularization term gives
\begin{equation}
\begin{aligned}
&c_{\bm\alpha}\sum_{m=0}^{M-1}
\|\kvec^{(m+1)}-\kvec^{(m)}\|^2
\le\frac12\sum_{i=1}^{N}\lambda_i(k_i^{(0)})^2\\
&\quad+\frac{2\sum_{t=1}^{T}\rho^{2t}}
{1-\rho^{2(2T-1)}}
(\bm{\tilde{v}}_0^{(0)})^\top\X^{-1}\bm{\tilde{v}}_0^{(0)}\\
&\quad+\frac{C_e^2r_p^2}
{2c_{\bm\alpha}(1-\rho^{2(2T-1)})}
(\bm{\tilde{v}}_0^{(0)})^\top\X^{-1}\bm{\tilde{v}}_0^{(0)}\\
&\quad+\frac{2MT}{\lambda_{\min}(\X)(1-\rho)^2}r_p^2\\
&\quad+\frac{MC_e^2}
{2c_{\bm\alpha}\lambda_{\min}(\X)(1-\rho)^2}r_p^4.
\end{aligned}
\label{eq:app:cumulative_gain_bound}
\end{equation}

The terms not multiplied by $M$ are finite and independent of
$M$. Dividing~\eqref{eq:app:cumulative_gain_bound} by
$c_{\bm\alpha}M$ and taking the limit superior therefore proves
\eqref{eq:average_gain_updates} with
\begin{equation}
\begin{aligned}
C_{\bm\alpha,1}
&:=\frac{2T}{c_{\bm\alpha}\lambda_{\min}(\X)(1-\rho)^2},\\
C_{\bm\alpha,2}
&:=\frac{C_e^2}{2c_{\bm\alpha}^2\lambda_{\min}(\X)(1-\rho)^2}.
\end{aligned}
\label{eq:app:gain_update_constants}
\end{equation}

When $r_p=0$, the perturbations and direction errors vanish.
The right-hand side of~\eqref{eq:app:cumulative_gain_bound} is
then independent of $M$. Letting $M\to\infty$ before dividing
by $M$ gives
\[
\begin{aligned}
&\sum_{m=0}^{\infty}\|\kvec^{(m+1)}-\kvec^{(m)}\|^2
\le\frac1{2c_{\bm\alpha}}\sum_{i=1}^{N}\lambda_i(k_i^{(0)})^2\\
&\quad+\frac{2\sum_{t=1}^{T}\rho^{2t}}
{c_{\bm\alpha}(1-\rho^{2(2T-1)})}
(\bm{\tilde{v}}_0^{(0)})^\top\X^{-1}\bm{\tilde{v}}_0^{(0)}
<\infty.
\end{aligned}
\]
Consequently, $\|\kvec^{(m+1)}-\kvec^{(m)}\|\to0$.
\end{proof}

\end{document}